\documentclass[12pt]{article}
\usepackage[latin1]{inputenc}
\usepackage[british]{babel}
\usepackage{cmap}
\usepackage{lmodern}

\usepackage{amssymb, amsmath, amsthm}
\usepackage[a4paper,top=25mm,bottom=25mm,left=25mm,right=25mm]{geometry}
\usepackage{ragged2e}

\usepackage{authblk} 
\usepackage{pifont}
\usepackage{graphicx}
\usepackage[usenames,dvipsnames,svgnames,table]{xcolor}
\usepackage[figuresright]{rotating}
\usepackage{xtab} 
\usepackage{longtable} 
\usepackage{multirow}
\usepackage{footnote}
\usepackage[stable]{footmisc}
\usepackage{chngpage} 
\usepackage{pdflscape} 
\usepackage[nottoc,notlot,notlof]{tocbibind} 

\usepackage{pgfplots}
\pgfplotsset{compat=1.14}
\pgfplotsset{every tick label/.append style={font=\footnotesize}}
\usepgfplotslibrary{fillbetween}
\usepackage{setspace}

\makesavenoteenv{tabular}
\usepackage{tabularx}
\usepackage{booktabs}
\usepackage{threeparttable} 
\usepackage[referable]{threeparttablex} 
\newcolumntype{R}{>{\raggedleft\arraybackslash}X}
\newcolumntype{L}{>{\raggedright\arraybackslash}X}
\newcolumntype{C}{>{\centering\arraybackslash}X}
\newcolumntype{M}[1]{>{\centering\arraybackslash}m{#1}}
\newcolumntype{K}{>{\columncolor{gray!20}}C}
\newcolumntype{k}{>{\columncolor{gray!20}}c}

\usepackage{dcolumn} 
\newcolumntype{.}{D{.}{.}{-1}}

\usepackage{tikz}
\usetikzlibrary{arrows, calc, matrix, patterns, positioning, trees}
\usepackage[semicolon]{natbib}
\usepackage[hyphens]{url}
\usepackage{hyperref} 
\hypersetup{
  colorlinks   = true,    
  urlcolor     = blue,    
  linkcolor    = blue,    
  citecolor    = ForestGreen      
}
\usepackage{microtype}
\usepackage[justification=centerfirst]{caption} 

\usepackage[labelformat=simple]{subcaption}

\DeclareCaptionLabelFormat{parenthesis}{(#2)}
\makeatletter
\renewcommand\p@subfigure{\arabic{figure}.}
\makeatother

\DeclareCaptionLabelFormat{parenthesis}{(#2)}
\makeatletter
\renewcommand\p@subtable{\arabic{table}.}
\makeatother

\usepackage{enumitem}
\setlist[itemize]{leftmargin=2.5\parindent}
\setlist[enumerate]{leftmargin=2.5\parindent}

\def\addlegendimage{\csname pgfplots@addlegendimage\endcsname}

\theoremstyle{plain}

\newtheorem{proposition}{Proposition}

\theoremstyle{definition}

\newtheorem{definition}{Definition}
\newtheorem{example}{Example}

\theoremstyle{remark}

\newcommand{\down}{\textcolor{BrickRed}{\rotatebox[origin=c]{270}{\ding{212}}}}
\newcommand{\up}{\textcolor{PineGreen}{\rotatebox[origin=c]{90}{\ding{212}}}}

\def\keywords{\vspace{.5em} 
{\noindent \textit{Keywords}: }}

\def\JEL{\vspace{.5em} 
{\noindent \textbf{\emph{JEL} classification number}: }}

\def\AMS{\vspace{.5em} 
{\noindent \textbf{\emph{MSC} class}: }}

\author{\href{https://sites.google.com/view/laszlocsato}{L\'aszl\'o Csat\'o}\thanks{~E-mail: \emph{laszlo.csato@sztaki.hu}} }
\affil{Institute for Computer Science and Control (SZTAKI) \\
E\"otv\"os Lor\'and Research Network (ELKH) \\
Laboratory on Engineering and Management Intelligence \\
Research Group of Operations Research and Decision Systems}
\affil{Corvinus University of Budapest (BCE) \\
Department of Operations Research and Actuarial Sciences}
\affil{Budapest, Hungary}

\title{A paradox of tournament seeding}
\date{\today}

\def\Dedication{ 
{\noindent
$\mathfrak{Gro{\ss}e}$ $\mathfrak{Beispiele}$ $\mathfrak{sind}$ $\mathfrak{die}$ $\mathfrak{besten}$ $\mathfrak{Lehrmeister}$, $\mathfrak{aber}$ $\mathfrak{freilich}$ $\mathfrak{ist}$ $\mathfrak{es}$ $\mathfrak{schlimm}$, $\mathfrak{wenn}$ $\mathfrak{sich}$ $\mathfrak{eine}$ $\mathfrak{Wolke}$ $\mathfrak{von}$ $\mathfrak{theoretischen}$ $\mathfrak{Vorurteilen}$ $\mathfrak{dazwischenlegt}$, $\mathfrak{denn}$ $\mathfrak{auch}$ $\mathfrak{das}$ $\mathfrak{Sonnenlicht}$ $\mathfrak{bricht}$ $\mathfrak{und}$ $\mathfrak{f\ddot{a}rbt}$ $\mathfrak{sich}$ $\mathfrak{in}$ $\mathfrak{Wolken}$. $\mathfrak{Solche}$ $\mathfrak{Vorurteile}$, $\mathfrak{die}$ $\mathfrak{sich}$ $\mathfrak{in}$ $\mathfrak{mancher}$ $\mathfrak{Zeit}$ $\mathfrak{wie}$ $\mathfrak{ein}$ $\mathfrak{Miasma}$ $\mathfrak{bilden}$ $\mathfrak{und}$ $\mathfrak{verbreiten}$, $\mathfrak{zu}$ $\mathfrak{zerst\ddot{o}ren}$, $\mathfrak{ist}$ $\mathfrak{eine}$ $\mathfrak{dringende}$ $\mathfrak{Pflicht}$ $\mathfrak{der}$ $\mathfrak{Theorie}$, $\mathfrak{denn}$ $\mathfrak{was}$ $\mathfrak{menschlicher}$ $\mathfrak{Verstand}$ $\mathfrak{f\ddot{a}lschlich}$ $\mathfrak{erzeugt}$, $\mathfrak{kann}$ $\mathfrak{auch}$ $\mathfrak{blo{\ss}er}$ $\mathfrak{Verstand}$ $\mathfrak{wieder}$ $\mathfrak{vernichten}$.
\footnote{~``\emph{Great examples are the best teachers, but it is certainly a misfortune if a cloud of theoretical prejudices comes between, for even the sunbeam is refracted and tinted by the clouds. To destroy such prejudices, which many a time rise and spread themselves like a miasma, is an imperative duty of theory, for the misbegotten offspring of human reason can also be in turn destroyed by pure reason.}''
(Source: Carl von Clausewitz: \emph{On War}, Book 4, Chapter 11 [The Battle---Continuation: The Use of the Battle], translated by Colonel James John Graham, London, N. Tr\"ubner, 1873. \url{http://clausewitz.com/readings/OnWar1873/TOC.htm})}}

\flushright
\noindent (Carl von Clausewitz: \emph{Vom Kriege})

\vspace{1cm} 
\justify }

\begin{document}
\newgeometry{top=20mm,bottom=20mm,left=25mm,right=25mm}

\maketitle
\thispagestyle{empty}
\Dedication

\begin{abstract}
\noindent
We analyse a mathematical model of seeding for sports contests with round-robin qualifying tournaments. The standard seeding system based on coefficients measuring the historical performance of the teams is shown to be unfair as it might potentially punish a team for its better results by having to face stronger opponents on average in the next stage. Major football competitions are revealed to suffer from this weakness.
Incentive compatibility can be guaranteed by providing each qualified team with the highest coefficient of all teams that are ranked lower in its qualifying tournament for seeding purposes. Our proposal is illustrated by the 2020/21 UEFA Champions League.



\AMS{62F07, 90B90, 91B14}

\JEL{C44, D71, Z20}
\end{abstract}

\clearpage
\newgeometry{top=25mm,bottom=25mm,left=25mm,right=25mm}

\section{Introduction} \label{Sec1}

Every sports tournament has to provide appropriate incentives for the contestants to exert effort \citep{Szymanski2003}. In particular, the ranking method should not reward teams for poor performance \citep{VaziriDabadghaoYihMorin2018}. But the rules are complex in practice and sometimes lead to unforeseen consequences such as a situation when a team is worse off by winning \citep{KendallLenten2017}.

Unsurprisingly, various theoretical models of sports contests have been considered in view of \emph{incentive compatibility}.
\citet{Pauly2014} derives an impossibility theorem for championships consisting of two qualifying tournaments with disjoint sets of participants. \citet{Vong2017} proves that, if more than one contestants advances to the next round, some players can benefit from shirking to qualify with a lower rank.
\citet{DagaevSonin2018} investigate tournament systems, composed of multiple round-robin and knockout tournaments with the same set of participants when the sets of winners of noncumulative prizes have a nonempty intersection. \citet{Csato2020c} considers group-based qualification systems where teams from different groups are compared, which can create incentives for both teams to play a draw instead of winning \citep{Csato2020d}. \citet{Csato2021a} and \citet{Csato2022a} present how neglecting these theoretical findings has led to problems in European football.
\citet{KrumerMegidishSela2020b} show that strategic considerations may motivate a contestant to lose in a round-robin tournament because this can result in a higher expected payoff.
Last but not least, \citet{LentenKendall2021} collect similar compelling ideas within the academic literature that could yet be considered and adopted by sports administrators.

Although the round-robin format in which each team meets all the others is one of the most common sports tournaments, it requires a lot of time to organise. On the other hand, if the competitors can play only against a limited number of opponents, the set of matches should be chosen carefully. This can be achieved through \emph{seeding}, by ordering the entrants based on playing history and/or the judgement of experts to pair them according to their ranks.
The problem of seeding in knockout tournaments has been thoroughly explored in the literature, see e.g.\ \citet{Hwang1982, HorenRiezman1985, Schwenk2000, GrohMoldovanuSelaSunde2012, DagaevSuzdaltsev2018, ArlegiDimitrov2020, DellaCroceDragottoScatamacchia2022}. The seeding rules of the most prominent football competition, the FIFA World Cup have also got serious attention \citep{ScarfYusof2011, Guyon2015a, LalienaLopez2019, CeaDuranGuajardoSureSiebertZamorano2020, Csato2022e}. Similarly, several statistical papers have analysed the effect of seeding on tournament outcome \citep{MonksHusch2009, CoronaForrestTenaWiper2019, DagaevRudyak2019, EngistMerkusSchafmeister2021}.

However, no academic work has addressed the incentive compatibility of the seeding rules except for \citet{Csato2020a}, a paper revealing a unique shortcoming in the UEFA Champions League group stage draw that emerged from the 2015/16 season due to a misaligned way of filling vacant slots.
Our main result is more universal: traditional seeding systems based on exogenous measures of teams' strengths are generically incentive incompatible---but they can be made strategyproof in a straightforward way.

Admittedly, these potentially unfair seeding rules rarely create incentives to lose \emph{ex ante}. Analogously to the previous literature \citep{Csato2019c, Csato2020c, Csato2020a, DagaevSonin2018}, we can only provide hypothetical examples when a team is better off by losing. This is not surprising since these regimes are widely used in practice, thus, they probably would have been changed long ago if the probability of a tanking opportunity would be substantially higher. Nonetheless, we think that
(a) even the punishment for better results can be detrimental to the business model of the sports industry if the stakeholders realise this unfairness at the end of a competition; and
(b) the academic community can make an important service to the decision-makers and the society by highlighting every issue of potential problems around incentives.

Consequently, the current paper makes a contribution by calling attention to a shortcoming of standard seeding systems that are extensively used in major sports competitions. Therefore, the regulatory bodies will have an opportunity to consider the seriousness of this weakness, and to adopt our proposal for an incentive compatible seeding rule, or to take other measures to prevent tanking, for instance, by choosing a schedule that maximises competitiveness.

Our roadmap is as follows.
Section~\ref{Sec2} presents two real-world cases to highlight the issue. A mathematical model is given in Section~\ref{Sec3}. We provide a strategyproof seeding mechanism and summarise policy implications in Section~\ref{Sec4}. Finally, Section~\ref{Sec6} concludes.

\section{Real-world illustrations} \label{Sec2}

Let us see two motivating examples.

\begin{example} \label{Examp1}
Assume the following hypothetical modifications to real-world results in the \href{https://en.wikipedia.org/wiki/2018_FIFA_World_Cup_qualification}{2018 FIFA World Cup qualification}:
\begin{itemize}
\item
Wales vs.\ Republic of Ireland was 2-1 (instead of 0-1) on 9 October 2017 in \href{https://en.wikipedia.org/wiki/2018_FIFA_World_Cup_qualification_\%E2\%80\%93_UEFA_Group_D}{UEFA Group D}. Consequently, Wales would have had 20 and the Republic of Ireland would have had 16 points in that group, thus Wales would have advanced to the \href{https://en.wikipedia.org/wiki/2018_FIFA_World_Cup_qualification_\%E2\%80\%93_UEFA_Second_Round}{UEFA Second Round} due to having 14 points in the comparison of the runners-up. There Wales would have been in Pot 1 rather than Denmark, therefore the tie Wales vs.\ Denmark would have been possible (in fact, Denmark played against the Republic of Ireland). Suppose that Wales qualified for the World Cup instead of Denmark.
\item
The first leg of Sweden vs.\ Italy was 1-1 (instead of 1-0) on 10 November 2017 in the UEFA Second Round, hence Italy qualified for the World Cup.
\end{itemize}
In the \href{https://en.wikipedia.org/wiki/2018_FIFA_World_Cup_seeding}{draw for the 2018 FIFA World Cup}, the composition of the pots depended on the \href{https://www.fifa.com/fifa-world-ranking/ranking-table/men/rank/id11976/}{October 2017 FIFA World Ranking}. The only exception was the automatic assignment of the host---Russia---to Pot 1 besides the seven highest-ranked qualified teams.
Hence Uruguay ($17$th in the relevant FIFA ranking) would have been drawn from Pot 3 as among the best $16$ teams, only Chile would have not qualified in the above scenario (Wales was the $14$th and Italy was the $15$th in the FIFA World Ranking of October 2017).

\begin{table}[t!]
  \centering
  \caption{Pot composition in the hypothetical 2018 FIFA World Cup}
  \label{Table1}
\rowcolors{3}{}{gray!20} 
\begin{threeparttable}
    \begin{tabularx}{\textwidth}{LLLL} \toprule
    Pot 1 & Pot 2 & Pot 3 & Pot 4 \\ \bottomrule
    Russia (65) & Spain (8) & \textbf{Uruguay (17)} & Serbia (38) \\
    Germany (1) & \textbf{Peru (10)} & Iceland (21) & Nigeria (41) \\
    Brazil (2) & Switzerland (11) & Costa Rica (22) & Australia (43) \\
    Portugal (3) & England (12) & Sweden (25) & Japan (44) \\
    Argentina (4) & Colombia (13) & Tunisia (28) & Morocco (48) \\
    Belgium (5) & \emph{Wales (14)} & Egypt (30) & Panama (49) \\
    Poland (6) & \emph{Italy (15)} & Senegal (32) & South Korea (62) \\
    France (7) & Mexico (16) & Iran (34) & Saudi Arabia (63) \\ \bottomrule
    \end{tabularx}
    
\begin{tablenotes} \footnotesize
	\item The pots are determined by the FIFA World Ranking of October 2017, see the numbers in parenthesis.
	\item Russia is the top seed as host.
	\item Teams written in \emph{italics} qualified only in the hypothetical but feasible scenario of Example~\ref{Examp1}.
	\item Uruguay (17)---the top team in Pot 3---would have been drawn from Pot 2 due to losing against Paraguay (34) in the South American qualifiers of the 2018 FIFA World Cup since then either Paraguay or New Zealand (122) would have qualified for the World Cup instead of Peru (10). The national teams affected by this modification are written in \textbf{bold}.
\end{tablenotes}
\end{threeparttable} 
\end{table}

The allocation of the teams in the above scenario is given in Table~\ref{Table1}.
Consider what would have happened if the result of the match Paraguay ($34$) vs.\ Uruguay, played on 5 September 2017 in the \href{https://en.wikipedia.org/wiki/2018_FIFA_World_Cup_qualification_(CONMEBOL)}{South American qualifiers}, would have been 2-1 instead of 1-2. Then Uruguay would have remained second and Paraguay would have been fifth in this qualifying competition. Paraguay would have played against New Zealand ($122$) in the \href{https://en.wikipedia.org/wiki/2018_FIFA_World_Cup_qualification_(OFC\%E2\%80\%93CONMEBOL_play-off)}{OFC--CONMEBOL qualification play-off}, thus Peru ($10$) could not have qualified for the World Cup. Therefore, Uruguay would have been drawn from the stronger Pot 2 instead of Pot 3 merely due to its loss against Paraguay.
It probably means a substantial advantage: in the 2018 FIFA World Cup, seven and two teams advanced to the knockout stage from Pots 2 and 3, respectively.
\end{example}

Example~\ref{Examp1} contains a small sloppiness since we have not checked whether the October 2017 FIFA World Ranking would have been changed. However, this issue does not affect the potential case of incentive incompatibility. 

\begin{example} \label{Examp2}
In the \href{https://en.wikipedia.org/wiki/2020\%E2\%80\%9321_UEFA_Europa_League_group_stage#Draw}{draw for the 2020/21 UEFA Europa League group stage}, the composition of the pots was determined by the 2020 UEFA club coefficients, available at \url{https://kassiesa.net/uefa/data/method5/trank2020.html}.
Assume the following hypothetical modifications to real-world results in the \href{https://en.wikipedia.org/wiki/2020\%E2\%80\%9321_UEFA_Europa_League_qualifying_phase_and_play-off_round\#Play-off_round}{play-off round of the qualifying phase}, with the first favourite team advancing to the group stage in place of the second unseeded underdog (the UEFA club coefficients are given in parenthesis):
\begin{itemize}
\item
Viktoria Plze{\v n} ($34.0$) against Hapoel Be'er Sheva ($14.0$);
\item
Basel ($58.5$) against CSKA Sofia ($4.0$);
\item
Sporting CP ($50.0$) against LASK ($14.0$);
\item
Copenhagen ($42.0$) against Rijeka ($11.0$);
\item
VfL Wolfsburg ($36.0$) against AEK Athens ($16.5$).
\end{itemize}
There were $48$ teams in the group stage. Leicester City ($22.0$) was ranked $20$th because Rapid Wien ($22.0$) had the same 2020 UEFA club coefficient but the tie-breaking criterion---coefficient in the next most recent season in which they are not equal \citep[Annex~D.8]{UEFA2020b}---preferred the latter club. Due to the above changes, five teams with a higher coefficient than Leicester City would have qualified instead of five teams with a lower coefficient. Hence, Leicester City would have been only the $25$th highest-ranked, namely, the best club in Pot 3 as each of the four pot consists of $12$ clubs.
In addition, suppose that Leicester defeated Norwich City at home by 2-1 (instead of 0-0) in the \href{https://en.wikipedia.org/wiki/2019\%E2\%80\%9320_Premier_League}{2019/20 English Premier League}. Then Leicester City would have remained fifth at the end of the season with $64$ points.

\begin{table}[t!]
  \centering
  \caption{Pot composition in the hypothetical 2020/21 UEFA Europa League}
  \label{Table2}
\rowcolors{3}{}{gray!20}
    \begin{tabularx}{1\textwidth}{LL} \toprule
    Pot 1 & Pot 2 \\ \bottomrule
    Arsenal (91.0) & Gent (39.5) \\
    \textbf{Tottenham Hotspur (85.0)} & PSV Eindhoven (37.0) \\
    Roma (80.0) & \emph{VfL Wolfsburg (36.0)} \\
    Napoli (77.0) & Celtic (34.0) \\
    Benfica (70.0) & \emph{Viktoria Plze{\v n} (34.0)} \\
    Bayer Leverkusen (61.0) & Dinamo Zagreb (33.5) \\
    \emph{Basel (58.5)} & Sparta Prague (30.5) \\
    Villarreal (56.0) & Slavia Prague (27.5) \\
    \emph{Sporting CP (50.0)} & Ludogorets Razgrad (26.0) \\
    CSKA Moscow (44.0) & Young Boys (25.5) \\
    \emph{Copenhagen (42.0)} & Crvena Zvezda (22.75) \\
    Braga (41.0) & Rapid Wien (22.0) \\ \bottomrule
    \end{tabularx} 

\vspace{0.25cm}
\begin{threeparttable}	
\rowcolors{3}{}{gray!20}
    \begin{tabularx}{1\textwidth}{LL} \toprule
    Pot 3 & Pot 4 \\ \bottomrule    
    \textbf{Leicester City (22.0)} & 1899 Hoffenheim (14.956) \\
    PAOK (21.0) & CFR Cluj (12.5) \\
    Qaraba{\u g} (21.0) & Zorya Luhansk (12.5) \\
    Standard Li{\` e}ge (20.5) & Nice (11.849) \\
    Real Sociedad (20.456) & Lille (11.849) \\
    Granada (20.456) & Dundalk (8.5) \\
    Milan (19.0) & Slovan Liberec (8.0) \\
    AZ Alkmaar (18.5) & Antwerp (7.58) \\
    Feyenoord (17.0) & Lech Poznan (7.0) \\
    Maccabi Tel Aviv (16.5) & Sivasspor (6.72) \\
    Rangers (16.25) & Wolfsberger AC (6.585) \\
    Molde (15.0) & Omonia (5.35) \\ \bottomrule
    \end{tabularx}
    
\begin{tablenotes} \footnotesize
	\item The pots are determined by the 2020 UEFA club coefficients, shown in parenthesis.
	\item Teams written in \emph{italics} qualified only in the hypothetical but feasible scenario of Example~\ref{Examp2}.
	\item Leicester City (22.0)---the top team in Pot 3---could have been drawn from Pot 2 due to losing against Wolverhampton (18.092) in the 2019/20 English Premier League since then the latter team could have qualified for the UEFA Europa League group stage instead of Tottenham Hotspur (85.0). The clubs affected by this modification are written in \textbf{bold}.
\end{tablenotes}
\end{threeparttable}
\end{table}

The allocation of the clubs in the above scenario is presented in Table~\ref{Table2}.
Consider what would have happened if the outcome of the match Wolverhampton Wanderers ($18.092$) vs.\ Leicester City, played on 14 February 2020 in the 2019/20 Premier League, would have been 1-0 instead of 0-0. Leicester would have remained fifth with $62$ points, while Wolverhampton would have been sixth with $61$ points rather than Tottenham Hotspur ($85.0$), which scored $59$ points. Consequently, Wolverhampton would have entered the Europa League qualification in the second qualifying round, and it could have qualified for the group stage in the place of Tottenham. Then Leicester would have been drawn from the stronger Pot 2 merely due to its loss against Wolverhampton.
This probably means an advantage, although in the 2020/21 Europa League, six teams advanced to the knockout stage from both Pots 2 and 3, respectively.
\end{example}

Examples~\ref{Examp1} and \ref{Examp2} uncover that a team can be worse off by winning in the South American qualifiers of the FIFA World Cup and the English Premier League, respectively. By changing the schedule of these round-robin tournaments such that the sensitive games (Paraguay vs.\ Uruguay and Wolverhampton Wanderers vs.\ Leicester City, respectively) are played in the last round, losing can be made beneficial for a team because its final ranking is not affected but the modified set of the teams qualifying (Paraguay instead of Peru and Wolverhampton Wanderers instead of Tottenham Hotspur, respectively) allows to be placed with a higher probability in a higher-ranked seeding pot, independently of the outcome of the matches to be played later. For instance, note that Leicester City always prefers to qualify together with Wolverhampton Wanderers instead of Tottenham Hotspur as the former scenario results in gaining one position in the ranking used for seeding compared to the latter scenario, which might lead to being in a better seeding pot with a positive probability.

\section{Theoretical background} \label{Sec3}

Consider a round-robin qualifying tournament with a set of teams $T$, where each team $t \in T$ has a coefficient $\xi_t$. The teams ranked between the $p$th and $q$th ($p \leq q$) positions qualify for the second stage. There, a qualified team $t$ has a seeding value $\Psi_t$, which is usually (but not necessarily) its coefficient $\xi_t$ as we have seen in Examples~\ref{Examp1} and \ref{Examp2}.

Any team prefers if more teams play in the second round with a lower coefficient.
In other words, it is a common belief that these measures positively correlate with true abilities. All coefficients used in practice are constructed along this line. For instance, the \href{https://en.wikipedia.org/wiki/FIFA_World_Rankings}{FIFA World Ranking} and the \href{https://en.wikipedia.org/wiki/UEFA_coefficient}{UEFA coefficients} for national teams, countries, and clubs alike award more points for wins (draws) than for draws (losses), thus better achievements in the past translate into a higher value.

Therefore, the first goal for every team is to qualify and the second goal is to qualify together with teams having a lower coefficient.
Any team $t \in T$ may lose a match to improve the second objective without deteriorating the first. Denote the (strict) rankings of the qualifying tournament by $\succ$ and $\succ'$, the ranks of team $s \in T$ by $\#(s)$ and $\#'(s)$, as well as the sets of teams qualified by
\[
Q = \{ s \in T: p \leq \#(s) \leq r \} \text{ and}
\]
\[
Q' = \{ s \in T: p \leq \#'(s) \leq r \}
\]
before and after the manipulation, respectively.

\begin{definition} \label{Def1}
\emph{Incentive compatibility with respect to seeding}:
A round-robin qualifying tournament is said to be \emph{incentive compatible with respect to seeding} if no team $t \in T$ \emph{ever} has a manipulation strategy such that:
\begin{itemize}
\item
team $t$ qualifies for the second stage both before and after the manipulation, that is, $t \in Q$ and $t \in Q'$;
\item
team $t$ has a better seeding position after the manipulation than before, namely, $|s \in Q': \Psi_s > \Psi_t| < |s \in Q: \Psi_s > \Psi_t|$.
\end{itemize}
Otherwise, the qualifying tournament is called \emph{incentive incompatible with respect to seeding}.
\end{definition}

As it has been revealed in Section~\ref{Sec2}, a qualifying tournament is incentive incompatible if $\Psi_t = \xi_t$ for all $t \in Q$ and $2 \leq |Q| < |T|$. In particular, a situation may exist where team $i$ has already secured qualification, while teams $j$ and $k$ compete for another slot such that $\xi_k > \xi_i > \xi_j$. Then team $i$ may consider losing against team $j$ in order to push it to the next stage at the expense of team $k$ as team $i$ can get a better seeding pot by taking team $j$ to the second round instead of team $k$.

The result below provides \emph{sufficient} conditions to prevent a strategic manipulation of this type. For the sake of simplicity, we assume that ties in the seeding values of qualified teams are broken in favour of the team ranked higher in the qualifying tournament.

\begin{proposition} \label{Prop1}
A round-robin qualifying tournament is incentive compatible with respect to seeding if at least one of the following conditions hold:
\begin{itemize}
\item
Only one team is allowed to qualify: $p = q$;
\item
All teams qualify: $p = 1$ and $q = |T|$;
\item
The seeding value of each qualified team $t \in Q$ in the second stage is equal to the maximal coefficient of the teams that are ranked lower than team $t$ in the qualifying round-robin tournament: $\Psi_t = \max \{ \xi_s: t \succ s \}$ for all $t \in Q$.
\end{itemize}
\end{proposition}

\begin{proof}
If $p = q$, then $t \in Q$ leads to $|s \in Q: \Psi_s > \Psi_t| = 0$. Consequently, no manipulation strategy can be found according to Definition~\ref{Def1}. \\
$p = 1$ and $q = |T|$ result in $Q = T$, thus $|s \in Q': \Psi_s > \Psi_t| = |s \in Q: \Psi_s > \Psi_t|$, which excludes the existence of a manipulation strategy as required by Definition~\ref{Def1}. \\
$\Psi_t = \max \{ \xi_s: t \succ s \}$ for all $t \in Q$ implies that
\[
s \in Q \text{ and } \Psi_s > \Psi_t \iff s \in Q \text{ and } s \succ t.
\]
Since team $t$ cannot be ranked higher in the round-robin qualifying tournament after the manipulation,
\[
|s \in Q': \Psi_s > \Psi_t| = |s \in Q': s \succ t| \geq |s \in Q: s \succ t| = |s \in Q: \Psi_s > \Psi_t|
\]
holds. Hence, the qualifying tournament is incentive compatible with respect to seeding.
\end{proof}

In the previous literature on the incentive compatibility of sports tournaments, only two papers consider multistage competitions. \citet{Csato2022d} discusses the case when the subsequent phases are related such that the results of some matches played in the previous stage(s) are carried over. \citet{Vong2017} analyses a general model with an arbitrary design but, crucially, the allocation of the teams into groups is assumed to be fixed \emph{a priori}. Therefore, in the proof of his main theorem, a team is interested in manipulation because it knows that if it achieves a higher rank in the qualifying tournament, then it will play in a group where elimination is guaranteed. On the other hand, our model contains more uncertainty regarding the opponents in the next stage. Hence, losing is beneficial only if it leads to a better seeding position, which is impossible under the conditions of Proposition~\ref{Prop1}.


\section{Discussion} \label{Sec4}

Now a general procedure is presented to guarantee our requirement, incentive compatibility with respect to seeding, on the basis of the theoretical model. Some alternative ideas are also outlined shortly.

According to Proposition~\ref{Prop1}, there are three ways to achieve strategyproofness in a round-robin qualifying tournament. However, the first two conditions---when exactly one team qualifies or all teams qualify for the next round---could not offer a universal rule. Nonetheless, they can be exploited in certain cases, for example, only one team advanced from the \href{https://en.wikipedia.org/wiki/2022_FIFA_World_Cup_qualification_(OFC)}{Oceanian (OFC) section of the 2022 FIFA World Cup qualification}.
Fortunately, there is a third opportunity, that is, to calculate the seeding value of any qualified team $t \in Q$ as $\Psi_t = \max \{ \xi_s: t \succ s \}$. In other words, team $t$ is seeded in the second stage based on the maximum of coefficients $\xi_s$ of all teams $s$ ranked lower than team $t$ in its round-robin qualifying competition. This is a reasonable rule: if team $i$ finishes ahead of team $j$ in a league, why is it judged worse for the draw in the next round?
Our proposal is called \emph{strategyproof seeding}.

\begin{sidewaystable}
  \centering
  \caption{Alternative rules for the draw of the 2020/21 UEFA Champions League group stage}  
  \label{Table3}
  \renewcommand\arraystretch{0.8}
\begin{threeparttable}
\rowcolors{1}{gray!20}{}
    \begin{tabularx}{\textwidth}{lccCCcCCc} \toprule \hiderowcolors
    Club  & Country & Position & Coefficient & Seeding & Inherited from & \multicolumn{3}{c}{Pot} \\
          &       &       &       & value &       & Official & Proposed & Change \\ \bottomrule \showrowcolors
    Bayern Munich & \multicolumn{2}{c}{CL TH (Germany 1st)} & 136   & 136   & ---   & 1 (1) & 1 (1) & --- \\
    Sevilla & \multicolumn{2}{c}{EL TH (Spain 4th)} & 102   & 102   & ---   & 1 (1) & 1 (1) & --- \\
    Real Madrid & Spain & 1st   & 134   & 134   & ---   & 1 (1) & 1 (1) & --- \\
    Liverpool & England & 1st   & 99    & 116   & Manchester City (2nd) & 1 (1) & 1 (1) & --- \\
    Juventus & Italy & 1st   & 117   & 117   & ---   & 1 (1) & 1 (1) & --- \\
    Paris Saint-Germain & France & 1st   & 113   & 113   & ---   & 1 (1) & 1 (1) & --- \\
    Zenit Saint Petersburg & Russia & 1st   & 64    & 64    & ---   & 1 (1) & 1 (1) & --- \\
    Porto & Portugal & 1st   & 75    & 75    & ---   & 1 (2) & 1 (3) & --- \\
    Barcelona & Spain & 2nd   & 128   & 128   & ---   & 2 (2) & 2 (2) & --- \\
    Atl\'etico Madrid & Spain & 3rd   & 127   & 127   & ---   & 2 (2) & 2 (2) & --- \\
    Manchester City & England & 2nd   & 116   & 116   & ---   & 2 (2) & 2 (2) & --- \\
    Manchester United & England & 3rd   & 100   & 100   & ---   & 2 (2) & 2 (2) & --- \\
    Shakhtar Donetsk & Ukraine & 1st   & 85    & 85    & ---   & 2 (2) & 2 (2) & --- \\
    Borussia Dortmund & Germany & 2nd   & 85    & 85    & ---   & 2 (1) & 2 (1) & --- \\
    Chelsea & England & 4th   & 83    & 91    & Arsenal (8th) & 2 (2) & 2 (2) & --- \\
    Ajax  & Netherlands & 1st   & 69.5  & 69.5  & ---   & 2 (2) & 3 (3) & \down \\
    Dynamo Kyiv & Ukraine & 2nd   & 55    & 55    & ---   & 3 (3) & 3 (3) & --- \\
    Red Bull Salzburg & Austria & 1st   & 53.5  & 53.5  & ---   & 3 (3) & 4 (4) & \down \\
    RB Leipzig & Germany & 3rd   & 49    & 61    & Bayer Leverkusen (5th) & 3 (3) & 3 (3) & --- \\
    Inter Milan & Italy & 2nd   & 44    & 80    & Roma (5th) & 3 (3) & 3 (3) & --- \\
    Olympiacos & Greece & 1st   & 43    & 43    & ---   & 3 (3) & 4 (4) & \down \\
    Lazio & Italy & 4th   & 41    & 80    & Roma (5th) & 3 (3) & 3 (3) & --- \\
    Krasnodar & Russia & 3rd   & 35.5  & 44    & CSKA Moscow (4th) & 3 (3) & 4 (4) & \down \\
    Atalanta & Italy & 3rd   & 33.5  & 80    & Roma (5th) & 3 (3) & 3 (3) & --- \\
    Lokomotiv Moscow & Russia & 2nd   & 33    & 44    & CSKA Moscow (4th) & 4 (4) & 4 (4) & --- \\
    Marseille & France & 2nd   & 31    & 83    & Lyon (7th) & 4 (4) & 2 (2) & \up \up \\
    Club Brugge & Belgium & 1st   & 28.5  & 39.5  & Gent (2nd) & 4 (4) & 4 (4) & --- \\
    Borussia M\"onchengladbach & Germany & 4th   & 26    & 61    & Bayer Leverkusen (5th) & 4 (4) & 3 (3) & \up \\
    Istanbul Ba{\c s}ak{\c s}ehir & Turkey & 1st   & 21.5  & 54    & Be{\c s}ikta{\c s} (3rd) & 4 (4) & 4 (4) & --- \\
    Midtjylland & Denmark & 1st   & 14.5  & 42    & Copenhagen (2nd) & 4 (4) & 4 (4) & --- \\
    Rennes & France & 3rd   & 14    & 83    & Lyon (7th) & 4 (4) & 3 (2) & \up \\
    Ferencv\'aros & Hungary & 1st   & 9     & 10.5  & Feh\'erv\'ar (2nd) & 4 (4) & 4 (4) & --- \\
    \bottomrule
    \end{tabularx}

\begin{tablenotes} \footnotesize
	\item CL (EL) TH stands for the UEFA Champions League (Europa League) titleholder.
	\item The column ``Inherited from'' shows the club of the domestic league whose UEFA club coefficient is taken over for seeding purposes.
	\item Proposed pot is the pot that contains the club if the current seeding policy applies to Pot 1. Since this rule is incentive compatible \citep{Csato2020a}, the pot according to the amendment suggested by \citet[Section~5]{Csato2020a} is reported in parenthesis for both the official and the strategyproof seeding systems.
	\item The column ``Change'' shows the movements of clubs between the pots due to the strategyproof seeding if the current seeding regime applies to Pot 1.
\end{tablenotes}
\end{threeparttable}    
\end{sidewaystable}

Table~\ref{Table3} applies strategyproof seeding for the 2020/21 Champions League group stage. Even though the seeding values of $15$ teams, including the $11$ lowest-ranked, are increased, it has only a moderated effect on the composition of pots as one German and two French teams benefit at the expense of four teams from the Netherlands, Austria, Greece, and Russia. That amendment usually favours the highest-ranked associations, where some clubs emerging without a robust European record (recall the unlikely triumph of Leicester City in the 2015/16 English Premier League \citep{BBC2016}) can ``obtain'' the performances of clubs with considerable achievements at the international level. Thus, strategyproof seeding contributes to the success of underdogs in the European cups, which may be advantageous for the long-run competitive balance in the top leagues. In addition, it probably better reflects the true abilities of the teams since playing more matches reduces the role of luck in sports tournaments \citep{McGarrySchutz1997, ScarfYusofBilbao2009, LasekGagolewski2018, SziklaiBiroCsato2022}. Consequently, it is more difficult to perform better in a round-robin league than in the Champions League or Europa League.

From the 2018/19 season onwards, UEFA club coefficients are determined either as the sum of all points won in the previous five years or as the association coefficient over the same period, \emph{whichever is the higher} \citep{Kassies2022a, UEFA2018g}. This rule was effective in the 2020/21 Champions League, the lower bound applied in the case of some Spanish, German, and French teams in the 2020/21 Europa League.
A somewhat similar policy is used in the UEFA Champions League and Europa League qualification, too, if a later round is drawn before the identity of the teams is known \citep{Csato2022b}: ``\emph{If, for any reason, any of the participants in such rounds are not known at the time of the draw, the coefficient of the club with the higher coefficient of the two clubs involved in an undecided tie is used for the purposes of the draw.}''
Therefore, the principle of strategyproof seeding is not unknown in UEFA club competitions, which can support its implementation.

Table~\ref{Table3} reinforces that the strategyproof seeding system may result in more ties than the current definition. If some teams inherit their seeding values from the same lower-ranked team, then these remain identical, and the tie should be broken by drawing of lots \citep[Annex~D.8]{UEFA2020a}. Although tie-breaking does not affect incentive compatibility, it is reasonable to prefer the teams ranked higher in the domestic league. If clubs from other associations also have the same seeding value (which has a much lower probability), they can be assigned arbitrarily in this equivalence class. Alternatively, the original club coefficients can be used for tie-breaking.
In Table~\ref{Table3}, two French teams at the boundary of Pots 2 and 3, Marseille and Rennes, inherit the same seeding value from Lyon. However, Marseille finished ahead of Rennes and has a higher coefficient, hence it is placed in Pot 2.

Our incentive compatible mechanism has further favourable implications. UEFA has modified the pot allocation policy in the Champions League from the 2015/16 season, probably inspired by the previous year when Manchester City, the English champion, was drawn from the second pot but Arsenal, the fourth-placed team in England, was drawn from the first pot. This decision---intended to strengthen the position of domestic titleholders \citep{UEFA2015e}---has considerable sporting effects \citep{CoronaForrestTenaWiper2019, DagaevRudyak2019}, especially since the poor way of filling vacancies leads to incentive incompatibility \citep{Csato2020a}. On the other hand, the proposed seeding rule guarantees that a national champion has at least the same seeding value as any team ranked lower in its domestic league.

Naturally, other strategyproof seeding policies can be devised. One example is the system of the 2020 UEFA European Championship: the ranking of all entrants on the basis of their results in the qualification. However, that principle is not appropriate if the achievements in the qualifying tournament(s) cannot be compared.
Another solution might be to associate seeding positions not with the coefficients but with the path of qualification. For instance, a club can be identified in the UEFA Champions League as the Spanish runner-up rather than by its name. The results of these ``labels'' can be measured by the achievements of the corresponding teams \citep{Guyon2015b}.

To conclude, the recommended strategyproof seeding mechanism provides incentive compatibility in any setting. While other rules are also able to eliminate perverse incentives, they are unlikely to be independent of the particular characteristics of the tournament.

\section{Conclusions} \label{Sec6}

The present work has analysed a mathematical model of seeding for sports tournaments where the teams qualify from round-robin competitions. Several contests are designed this way, including the most prestigious football tournaments (FIFA World Cup, UEFA European Championship, UEFA Champions League). The necessary conditions of incentive incompatibility have turned out to be quite restrictive: if each competitor is considered with its own coefficient (usually a measure of its past performance), only one or all of them should qualify from every round-robin contest.

Similar to the main findings of \citet{Vong2017} and \citet{KrumerMegidishSela2020b}, our result has the flavour of an impossibility theorem at first glance. However, here we can achieve strategyproofness by giving to each qualified competitor the highest coefficient of all competitors that are ranked lower in its round-robin qualifying tournament for seeding purposes.

The central message of this paper for decision makers is consonant with the conclusion of \citet{HaugenKrumer2021}, that is, tournament design should be included into the family of traditional topics discussed by sports management. In particular, administrators are strongly encouraged to follow our recommendation in order to prevent the occurrence of costly scandals in the future.

\section*{Acknowledgements}
\addcontentsline{toc}{section}{Acknowledgements}
\noindent
\emph{Ritxar Arlegi}, \emph{Marco Sahm}, and five anonymous reviewers provided valuable comments on earlier drafts. \\
We are indebted to the \href{https://en.wikipedia.org/wiki/Wikipedia_community}{Wikipedia community} for collecting and structuring valuable information on the sports tournaments discussed.

\bibliographystyle{apalike}
\bibliography{All_references}

\begin{thebibliography}{}

\bibitem[Arlegi and Dimitrov, 2020]{ArlegiDimitrov2020}
Arlegi, R. and Dimitrov, D. (2020).
\newblock Fair elimination-type competitions.
\newblock {\em European Journal of Operational Research}, 287(2):528--535.

\bibitem[{BBC}, 2016]{BBC2016}
{BBC} (2016).
\newblock Leicester {C}ity win {P}remier {L}eague title after {T}ottenham draw
  at {C}helsea.
\newblock 2 May. \url{https://www.bbc.com/sport/football/35988673}.

\bibitem[Cea et~al., 2020]{CeaDuranGuajardoSureSiebertZamorano2020}
Cea, S., Dur{\'a}n, G., Guajardo, M., Saur{\'e}, D., Siebert, J., and Zamorano,
  G. (2020).
\newblock An analytics approach to the {FIFA} ranking procedure and the {W}orld
  {C}up final draw.
\newblock {\em Annals of Operations Research}, 286(1-2):119--146.

\bibitem[Corona et~al., 2019]{CoronaForrestTenaWiper2019}
Corona, F., Forrest, D., Tena, J.~D., and Wiper, M. (2019).
\newblock Bayesian forecasting of {UEFA} {C}hampions {L}eague under alternative
  seeding regimes.
\newblock {\em International Journal of Forecasting}, 35(2):722--732.

\bibitem[Csat\'o, 2019]{Csato2019c}
Csat\'o, L. (2019).
\newblock {UEFA} {C}hampions {L}eague entry has not satisfied strategyproofness
  in three seasons.
\newblock {\em Journal of Sports Economics}, 20(7):975--981.

\bibitem[Csat\'o, 2020a]{Csato2020c}
Csat\'o, L. (2020a).
\newblock The incentive (in)compatibility of group-based qualification systems.
\newblock {\em International Journal of General Systems}, 49(4):374--399.

\bibitem[Csat\'o, 2020b]{Csato2020a}
Csat\'o, L. (2020b).
\newblock The {UEFA} {C}hampions {L}eague seeding is not strategy-proof since
  the 2015/16 season.
\newblock {\em Annals of Operations Research}, 292(1):161--169.

\bibitem[Csat\'o, 2020c]{Csato2020d}
Csat\'o, L. (2020c).
\newblock When neither team wants to win: {A} flaw of recent {UEFA}
  qualification rules.
\newblock {\em International Journal of Sports Science \& Coaching},
  15(4):526--532.

\bibitem[Csat\'o, 2021]{Csato2021a}
Csat\'o, L. (2021).
\newblock {\em Tournament Design: How Operations Research Can Improve Sports
  Rules}.
\newblock Palgrave Pivots in Sports Economics. Palgrave Macmillan, Cham,
  Switzerland.

\bibitem[Csat\'o, 2022a]{Csato2022e}
Csat\'o, L. (2022a).
\newblock Group draw with unknown qualified teams: A lesson from 2022 {FIFA}
  {W}orld {C}up.
\newblock {\em International Journal of Sports Science \& Coaching}, in press.
\newblock {DOI}:
  \href{https://doi.org/10.1177/17479541221108799}{10.1177/17479541221108799}.

\bibitem[Csat\'o, 2022b]{Csato2022d}
Csat\'o, L. (2022b).
\newblock How to design a multi-stage tournament when some results are carried
  over?
\newblock {\em OR Spectrum}, 44(3):683--707.

\bibitem[Csat\'o, 2022c]{Csato2022a}
Csat\'o, L. (2022c).
\newblock Quantifying incentive (in)compatibility: {A} case study from sports.
\newblock {\em European Journal of Operational Research}, 302(2):717--726.

\bibitem[Csat\'o, 2022d]{Csato2022b}
Csat\'o, L. (2022d).
\newblock {UEFA} against the champions? {A}n evaluation of the recent reform of
  the {C}hampions {L}eague qualification.
\newblock {\em Journal of Sports Economics}, 23(8):991--1016.

\bibitem[Dagaev and Rudyak, 2019]{DagaevRudyak2019}
Dagaev, D. and Rudyak, V. (2019).
\newblock Seeding the {UEFA} {C}hampions {L}eague participants: Evaluation of
  the reform.
\newblock {\em Journal of Quantitative Analysis in Sports}, 15(2):129--140.

\bibitem[Dagaev and Sonin, 2018]{DagaevSonin2018}
Dagaev, D. and Sonin, K. (2018).
\newblock Winning by losing: {I}ncentive incompatibility in multiple
  qualifiers.
\newblock {\em Journal of Sports Economics}, 19(8):1122--1146.

\bibitem[Dagaev and Suzdaltsev, 2018]{DagaevSuzdaltsev2018}
Dagaev, D. and Suzdaltsev, A. (2018).
\newblock Competitive intensity and quality maximizing seedings in knock-out
  tournaments.
\newblock {\em Journal of Combinatorial Optimization}, 35(1):170--188.

\bibitem[Della~Croce et~al., 2022]{DellaCroceDragottoScatamacchia2022}
Della~Croce, F., Dragotto, G., and Scatamacchia, R. (2022).
\newblock On fairness and diversification in {WTA} and {ATP} tennis tournaments
  generation.
\newblock {\em Annals of Operations Research}, 316(2):1107--1119.

\bibitem[Engist et~al., 2021]{EngistMerkusSchafmeister2021}
Engist, O., Merkus, E., and Schafmeister, F. (2021).
\newblock The effect of seeding on tournament outcomes: {E}vidence from a
  regression-discontinuity design.
\newblock {\em Journal of Sports Economics}, 22(1):115--136.

\bibitem[Groh et~al., 2012]{GrohMoldovanuSelaSunde2012}
Groh, C., Moldovanu, B., Sela, A., and Sunde, U. (2012).
\newblock Optimal seedings in elimination tournaments.
\newblock {\em Economic Theory}, 49(1):59--80.

\bibitem[Guyon, 2015a]{Guyon2015b}
Guyon, J. (2015a).
\newblock Champions {L}eague: {H}ow to {S}olve the {S}eeding {P}roblem.
\newblock {\em The New York Times}.
\newblock 21 January.
  \url{https://www.nytimes.com/2015/01/22/upshot/champions-league-how-to-solve-the-seeding-problem.html}.

\bibitem[Guyon, 2015b]{Guyon2015a}
Guyon, J. (2015b).
\newblock Rethinking the {FIFA} {W}orld {C}up\textsuperscript{{TM}} final draw.
\newblock {\em Journal of Quantitative Analysis in Sports}, 11(3):169--182.

\bibitem[Haugen and Krumer, 2021]{HaugenKrumer2021}
Haugen, K.~K. and Krumer, A. (2021).
\newblock On importance of tournament design in sports management: {E}vidence
  from the {UEFA} {E}uro 2020 qualification.
\newblock In Ratten, V., editor, {\em Innovation and Entrepreneurship in Sport
  Management}, pages 22--35. Edward Elgar Publishing, New York.

\bibitem[Horen and Riezman, 1985]{HorenRiezman1985}
Horen, J. and Riezman, R. (1985).
\newblock Comparing draws for single elimination tournaments.
\newblock {\em Operations Research}, 33(2):249--262.

\bibitem[Hwang, 1982]{Hwang1982}
Hwang, F.~K. (1982).
\newblock New concepts in seeding knockout tournaments.
\newblock {\em The American Mathematical Monthly}, 89(4):235--239.

\bibitem[Kassies, 2022]{Kassies2022a}
Kassies, B. (2022).
\newblock {UEFA} {C}oefficients calculation method.
\newblock \url{https://kassiesa.net/uefa/calc.html}.

\bibitem[Kendall and Lenten, 2017]{KendallLenten2017}
Kendall, G. and Lenten, L.~J.~A. (2017).
\newblock When sports rules go awry.
\newblock {\em European Journal of Operational Research}, 257(2):377--394.

\bibitem[Krumer et~al., 2020]{KrumerMegidishSela2020b}
Krumer, A., Megidish, R., and Sela, A. (2020).
\newblock Strategic manipulations in round-robin tournaments.
\newblock Manuscript. URL:
  \url{https://www.researchgate.net/publication/339254201}.

\bibitem[Laliena and L{\'o}pez, 2019]{LalienaLopez2019}
Laliena, P. and L{\'o}pez, F.~J. (2019).
\newblock Fair draws for group rounds in sport tournaments.
\newblock {\em International Transactions in Operational Research},
  26(2):439--457.

\bibitem[Lasek and Gagolewski, 2018]{LasekGagolewski2018}
Lasek, J. and Gagolewski, M. (2018).
\newblock The efficacy of league formats in ranking teams.
\newblock {\em Statistical Modelling}, 18(5-6):411--435.

\bibitem[Lenten and Kendall, 2021]{LentenKendall2021}
Lenten, L.~J.~A. and Kendall, G. (2021).
\newblock Scholarly sports: Influence of social science academe on sports rules
  and policy.
\newblock {\em Journal of the Operational Research Society}, in press.
\newblock {DOI}:
  \href{https://doi.org/10.1080/01605682.2021.2000896}{10.1080/01605682.2021.2000896}.

\bibitem[McGarry and Schutz, 1997]{McGarrySchutz1997}
McGarry, T. and Schutz, R.~W. (1997).
\newblock Efficacy of traditional sport tournament structures.
\newblock {\em Journal of the Operational Research Society}, 48(1):65--74.

\bibitem[Monks and Husch, 2009]{MonksHusch2009}
Monks, J. and Husch, J. (2009).
\newblock The impact of seeding, home continent, and hosting on {FIFA} {W}orld
  {C}up results.
\newblock {\em Journal of Sports Economics}, 10(4):391--408.

\bibitem[Pauly, 2014]{Pauly2014}
Pauly, M. (2014).
\newblock Can strategizing in round-robin subtournaments be avoided?
\newblock {\em Social Choice and Welfare}, 43(1):29--46.

\bibitem[Scarf et~al., 2009]{ScarfYusofBilbao2009}
Scarf, P., Yusof, M.~M., and Bilbao, M. (2009).
\newblock A numerical study of designs for sporting contests.
\newblock {\em European Journal of Operational Research}, 198(1):190--198.

\bibitem[Scarf and Yusof, 2011]{ScarfYusof2011}
Scarf, P.~A. and Yusof, M.~M. (2011).
\newblock A numerical study of tournament structure and seeding policy for the
  soccer {W}orld {C}up {F}inals.
\newblock {\em Statistica Neerlandica}, 65(1):43--57.

\bibitem[Schwenk, 2000]{Schwenk2000}
Schwenk, A.~J. (2000).
\newblock What is the correct way to seed a knockout tournament?
\newblock {\em The American Mathematical Monthly}, 107(2):140--150.

\bibitem[Sziklai et~al., 2022]{SziklaiBiroCsato2022}
Sziklai, B.~R., Bir\'o, P., and Csat{\'o}, L. (2022).
\newblock The efficacy of tournament designs.
\newblock {\em Computers \& Operations Research}, 144:105821.

\bibitem[Szymanski, 2003]{Szymanski2003}
Szymanski, S. (2003).
\newblock The economic design of sporting contests.
\newblock {\em Journal of Economic Literature}, 41(4):1137--1187.

\bibitem[UEFA, 2015]{UEFA2015e}
UEFA (2015).
\newblock Champions' bonus for group stage draw.
\newblock 24 April.
  \url{https://www.uefa.com/uefachampionsleague/news/0220-0e15f38d3608-4437365f0e2f-1000--champions-bonus-for-group-stage-draw/}.

\bibitem[UEFA, 2018]{UEFA2018g}
UEFA (2018).
\newblock How the club coefficients are calculated.
\newblock 1 July.
  \url{https://www.uefa.com/nationalassociations/uefarankings/news/0252-0cda38714c0d-0874ab234eb6-1000--how-the-club-coefficients-are-calculated/}.

\bibitem[UEFA, 2020a]{UEFA2020a}
UEFA (2020a).
\newblock {\em Regulations of the UEFA Champions League 2018-21 Cycle. 2020/21
  Season}.
\newblock
  \url{http://web.archive.org/web/20210603140626/https://documents.uefa.com/r/Regulations-of-the-UEFA-Champions-League-2020/21-Online}.

\bibitem[UEFA, 2020b]{UEFA2020b}
UEFA (2020b).
\newblock {\em Regulations of the UEFA Europa League 2018-21 Cycle. 2020/21
  Season}.
\newblock
  \url{http://web.archive.org/web/20210517095618/https://documents.uefa.com/r/Regulations-of-the-UEFA-Europa-League-2020/21-Online}.

\bibitem[Vaziri et~al., 2018]{VaziriDabadghaoYihMorin2018}
Vaziri, B., Dabadghao, S., Yih, Y., and Morin, T.~L. (2018).
\newblock Properties of sports ranking methods.
\newblock {\em Journal of the Operational Research Society}, 69(5):776--787.

\bibitem[Vong, 2017]{Vong2017}
Vong, A.~I.~K. (2017).
\newblock Strategic manipulation in tournament games.
\newblock {\em Games and Economic Behavior}, 102:562--567.

\end{thebibliography}

\end{document}